\documentclass[reqno,11 pt]{amsart}
\usepackage{a4wide}
\usepackage[english]{babel}
\usepackage[utf8]{inputenc}
\usepackage{amsmath,amsthm,amsfonts}
\usepackage{braket}
\usepackage{hyperref}
\usepackage{relsize}
\usepackage{tikz}

\newcommand{\erf}{\operatorname{erf}}
\renewcommand{\Re}{\operatorname{Re}}
\renewcommand{\Im}{\operatorname{Im}}
\newcommand{\DN}{{\hspace{-0.05cm}\mathsmaller{D},\mathsmaller{N}}}
\newcommand{\D}{{\hspace{-0.05cm}\mathsmaller{D}}}
\newcommand{\N}{{\hspace{-0.05cm}\mathsmaller{N}}}
\newcommand{\vect}[2]{(\begin{smallmatrix} #1 \\ #2 \end{smallmatrix})}
\newcommand{\vvect}[2]{\big(\begin{smallmatrix} #1 \\ #2 \end{smallmatrix}\big)}

\newtheorem{satz}{Theorem}[section]
\newtheorem{lem}[satz]{Lemma}
\newtheorem{prop}[satz]{Proposition}
\newtheorem{cor}[satz]{Corollary}
\newtheorem{defi}[satz]{Definition}
\newtheorem{bem}[satz]{Remark}
\numberwithin{equation}{section}

\begin{document}

\title[]{Infinite order differential operators
associated\\ with superoscillations in the half-plane barrier}
\author{Peter Schlosser}
\address{(PS) Dipartimento di Matematica, Politecnico Milano, Via E. Bonardi 9, Milano}
\email{pschlosser@math.tugraz.at}
\thanks{This research was funded by the Austrian Science Fund (FWF) under Grant No. J 4685-N and by the European Union -- NextGenerationEU}
\subjclass[2020]{35A20, 35A08}
\keywords{Superoscillations, Schrödinger equation, Green's function, Half-plane barrier}

\begin{abstract}
Superoscillations are a phenomenon in physics, where linear combinations of low-frequency plane waves interfere almost destructively in such a way that the resulting wave has a higher frequency than any of the individual waves. 
The evolution of superoscillatory initial datum under the time dependent Schrödinger equation is stable in free space, but in general it is unclear whether it can be preserved in the presence of an external potential. 
In this paper, we consider the two-dimensional problem of superoscillations interacting with a half-plane barrier, where homogenous Dirichlet or Neumann boundary conditions are imposed on the negative $x_2$-semiaxis. We use the Fresnel integral technique to write the wave function as an absolute convergent Green's function integral. Moreover, we introduce the propagator of the Schrödinger equation in form of an infinite order differential operator, acting continuously 
on the function space of exponentially bounded entire functions.
In particular, this operator allows to prove that the property of superoscillations is preserved in the form of a similar phenomenon called supershift, which is stable over time.
\end{abstract}

\maketitle

\section{Introduction}

The concept of superoscillations was first introduced in the context of antenna theory in the 1950s, see the paper \cite{T52}. However, it was in the 1990s that Y. Aharonov and his collaborators discovered the connection between superoscillations and quantum mechanics, specifically weak values, see  \cite{AAV88,APR91,AV90}, but also the later publications \cite{AB05,ACNSST12,AGJR16,APR21,APT10} dealing with several developments of the theory of superoscillations.

\medskip

A mathematical investigation of a quantum mechanical superoscillating wave or particle always reduces to the time dependent Schrödinger equation with some potential $V$ and a superoscillating function $F$ as initial condition:
\begin{align*}
i\frac{\partial}{\partial t}\Psi(t,\mathbf{x})&=\big(-\Delta+V(t,\mathbf{x})\big)\Psi(t,\mathbf{x}), && t>0,\,\mathbf{x}\in\Omega, \\
\Psi(0,\mathbf{x})&=F(\mathbf{x}), && \mathbf{x}\in\Omega.
\end{align*}
The question whether the solution $\Psi(t,\mathbf{x})$ is again superoscillating at times $t>0$ was first proven for free particles in \cite{ACSST13_1,ACSST13_2}, and later also for nonvanishing potentials as the harmonic oscillator in \cite{ACSS20,ACSS18,BCSS14,BS15}, the electric field in \cite{ACSST17,ACST18,ACSS18,BCSS14}, the magnetic field in \cite{ACST18,CGS17,CPSW23}, the centrifugal potential in \cite{ACST18,CGS19}, the step potential in \cite{ACSST20} and distributional potentials as $\delta$ and $\delta'$ in \cite{ABCS20,ABCS21}. A unified approach to those problems was given in \cite{ABCS22,S22}, where under certain assumptions on the corresponding Green's function, the time persistence property of superoscillations was investigated for whole classes of potentials. Another general approach was given in \cite{PW22} who provide conditions on the moments of the Green's function in order to obtain similar time persistence results.

\medskip

A shared characteristic among the aforementioned examples is that they solely focus on potentials within a single spatial dimension.
 There are very few publications which treat the time persistence problem of the Schrödinger equation in two or more dimensions; some of them are \cite{ACJSSST22,ACSST16,ACSST17,ACST18,CGS17}. In this paper we consider the two-dimensional half-plane barrier with Dirichlet (or Neumann) boundary conditions. In particular, we use the setting where the barrier is located on the negative $x_2$-semiaxis $\Gamma:=\{\vect{0}{x_2}\;|\;x_2\leq 0\}$, i.e., we consider the Schrödinger equation on $\Omega:=\mathbb{R}^2\setminus\Gamma$, namely

\medskip
\begin{flushleft}
\begin{minipage}{0.29\textwidth}
\begin{center}
\begin{tikzpicture}
\fill[black!30] (-1,-1)--(1,-1)--(1,1)--(-1,1);
\draw[thick,fill=white] (-0.05,-1)--(-0.05,0) arc (180:0:0.05)--(0.05,-1);
\draw[->] (0.05,0)--(1.3,0) node[anchor=south] {\small{$x_1$}};
\draw[->] (0,0.05)--(0,1.35) node[anchor=north west] {\small{$x_2$}};
\draw (-0.5,0.5) node[anchor=center] {\Large{$\Omega$}};
\draw (0,-0.5) node[anchor=west] {\large{$\Gamma$}};
\end{tikzpicture}
\end{center}
\end{minipage}
\begin{minipage}{0.7\textwidth}
\begin{align}
i\frac{\partial}{\partial t}\Psi_\DN(t,\mathbf{x})&=-\Delta\Psi_\DN(t,\mathbf{x}), && t>0,\,\mathbf{x}\in\Omega, \notag \\
\Psi_\D(t,\mathbf{x})=0\quad\Big(\text{or}&\;\;\frac{\partial\Psi_\N}{\partial x_1}(t,\mathbf{x})=0\Big), && t>0,\,\mathbf{x}\in\Gamma, \label{Eq_Schroedinger_equation} \\
\Psi_\DN(0,\mathbf{x})&=F(\mathbf{x}), && \mathbf{x}\in\Omega. \notag
\end{align}
\end{minipage}
\end{flushleft}

\medskip

\noindent Here and in the rest of the paper, the indices $D$ and $N$ always emphasize the particular boundary condition used (Dirichlet or Neumann).

\medskip

The strategy to solve the Schrödinger equation of the half-plane barrier is based on Green's functions techniques, i.e., we write the wave function $\Psi_\DN$ as an integral of the form
\begin{equation}\label{Eq_Psi_formal}
\Psi_\DN(t,\mathbf{x})=\int_\Omega G_\DN(t,\mathbf{x},\mathbf{y})F(\mathbf{y})d\mathbf{y},\quad t>0,\mathbf{x}\in\Omega.
\end{equation}
A particular class of initial conditions are entire functions in two complex variables which grow at most exponentially at $\infty$. More precisely, we consider the function space
\begin{equation*}
\mathcal{A}_1(\mathbb{C}^2):=\Set{F:\mathbb{C}^2\rightarrow\mathbb{C}\text{ entire} | \exists A,B\geq 0:\,|F(\mathbf{z})|\leq Ae^{B|\mathbf{z}|},\text{ for every }\mathbf{z}\in\mathbb{C}^2},
\end{equation*}
equipped with the following notion of convergence: A sequence $(F_n)_n\in\mathcal{A}_1(\mathbb{C}^2)$ converges to $F_0\in\mathcal{A}_1(\mathbb{C}^2)$ if and only if there exists some $B\geq 0$, such that
\begin{equation}\label{Eq_A1_convergence}
\lim\limits_{n\rightarrow\infty}\sup\limits_{z\in\mathbb{C}^2}|F_n(z)-F_0(z)|e^{-B|z|}=0.
\end{equation}
The standard examples of superoscillatory functions, which serve as initial condition in the above Schrödinger equation are of the form
\begin{equation}\label{Eq_Fn_example}
F_n(\mathbf{x})=\sum\limits_{j=0}^nC_j(n)e^{ik_j(n)^{p_1}x_1+ik_j(n)^{p_2}x_2},\quad\mathbf{x}=\vvect{x_1}{x_2}\in\mathbb{R}^2,
\end{equation}
for some $a>1$ and the coefficients are given by
$$
C_j(n)=\binom{n}{j}\Big(\frac{1+a}{2}\Big)^{n-j}\Big(\frac{1-a}{2}\Big)^j,
$$
and
$$
k_j(n)=1-\frac{2j}{n},
 $$
where $p_1,p_2\in\mathbb{N}$. The superoscillatory property of these functions comes from the fact that although that the frequencies $|k_j(n)|\leq 1$ are in modulus bounded by $1$, the sequence $(F_n)_n$ of functions converges as
\begin{equation*}
\lim\limits_{n\rightarrow\infty}F_n(\mathbf{x})=e^{ia^{p_1}x_1+ia^{p_2}x_2},\quad\mathbf{x}\in\mathbb{R}^2,
\end{equation*}
to a plane wave with frequency $a>1$. However, also other types of superoscillating functions were considered in the past by different physical and mathematical communities. An overview, but also one general definition, was given in the recent paper \cite{BCSS22}, which puts most of the existing notions of superoscillations into a common framework.
In the paper \cite{BCSS22} are considered superoscillations in one dimension,
the natural two-dimensional extension, which we will need in this paper, reads as follows:

\begin{defi}[Superoscillations]\label{defi_Superoscillations}
A sequence of functions of the form
\begin{equation}\label{Eq_Fn_integral}
F_n(\mathbf{z})=\int_{|\mathbf{k}|\leq k_0}e^{i\mathbf{kz}}d\mu_n(\mathbf{k}),\quad\mathbf{z}\in\mathbb{C}^2,
\end{equation}
with a common maximal frequency $k_0>0$ and complex Borel measures $\mu_n$ on the closed ball $\overline{B_{k_0}(\mathbf{0})}\subseteq\mathbb{R}^2$ of radius $k_0$, is called superoscillating, if there exists some $\mathbf{a}\in\mathbb{R}^2$ with $|\mathbf{a}|>k_0$, such that
\begin{equation}\label{Eq_Fn_convergence}
\lim\limits_{n\rightarrow\infty}F_n(\mathbf{z})=e^{i\mathbf{az}}\quad\text{in }\mathcal{A}_1(\mathbb{C}^2).
\end{equation}
Here the products $\mathbf{kz}:=k_1z_1+k_2z_2$ of real or complex vectors is understood in the usual bilinear sense.
\end{defi}

\begin{bem}
We point out that any function $F_n$ of the form \eqref{Eq_Fn_integral} is automatically contained in $\mathcal{A}_1(\mathbb{C}^2)$. The exponential bound is given by the estimate
\begin{equation*}
|F_n(\mathbf{z})|\leq\int_{|\mathbf{k}|\leq k_0}e^{|\mathbf{kz}|}d|\mu_n|(\mathbf{k})\leq|\mu_n|\big(\overline{U_{k_0}(\mathbf{0})}\big)e^{k_0|\mathbf{z}|},\quad\mathbf{z}\in\mathbb{C}^2,
\end{equation*}
with $|\mu_n|$ the total variation of the complex measure $\mu_n$. The holomorphicity also follows from this locally uniform upper bound and a version of the dominated convegence theorem, which allows to interchange derivative and integral and leads to a holomorphic function $F_n(\mathbf{z})$.
\end{bem}

\medskip
To study the evolution of superoscillating functions we will introduce for every $t>0$, $\mathbf{x}\in\Omega$ an infinite order differential operator of the form
\begin{equation*}
U_\DN(t,\mathbf{x})=\sum\limits_{n_1,n_2=0}^\infty c_{n_1,n_2}(t,\mathbf{x})\frac{\partial^{n_1+n_2}}{\partial z_1^{n_1}\partial z_2^{n_2}},
\end{equation*}
where, the coefficients $c_{n_1,n_2}(t,\mathbf{x})$ depend on the potential and on the auxiliary  complex variables $z_1$ and $z_2$.
These operators applied to the initial datum $F=F_n$ of \eqref{Eq_Schroedinger_equation}, analytically extended to a holomorphic function, give the solution $\Psi_\DN(t,\mathbf{x})$ as
\begin{equation*}
\Psi_\DN(t,\mathbf{x})=U_\DN(t,\mathbf{x})F(\mathbf{z})\Big|_{\mathbf{z}=\mathbf{0}}.
\end{equation*}
As we will see the continuity on spaces of entire functions 
of the above operator $U_\DN$ is of crucial importance in the investigation of the time evolution of superoscillations. 

\medskip
{\em Plan of the paper.}
In Section \ref{sec_The_Greens_function_of_the_half_plane_barrier}
we consider the Green's function of the half-plane barrier 
and the integral representation of the solution of
the Cauchy problem for the Schrodinger equation using Fresnel integrals.

In Section \ref{sec_The_time_evolution_operator}
we identify suitable  infinite order differential operators associated with  half-plane barrier that will be the key tools to study, in Section 
\ref{sec_Time_persistence_of_superoscillations},
the time persistence of superoscillations and the supershift
property of the solution of Schr\"odinger equation with superoscillatory initial datum.

\section{The Green's function of the half-plane barrier}\label{sec_The_Greens_function_of_the_half_plane_barrier}

The strategy to solve the Schrödinger equation of the half-plane barrier is based on Green's functions techniques, i.e., we write the solution of the Schrödinger
equation with initial condition $F$ as an integral of the form
\begin{equation}\label{Eq_Psi_formal}
\Psi(t,\mathbf{x})=\int_\Omega G(t,\mathbf{x},\mathbf{y})F(\mathbf{y})d\mathbf{y},\quad t>0,\,\mathbf{x}\in\Omega.
\end{equation}
The Green's function $G$ for the particular problem of the half-plane barrier is calculated in \cite{S82}, and using polar coordinates $\mathbf{x}=r\vect{\cos\varphi}{\sin\varphi}$, $\mathbf{y}=\vect{\cos\theta}{\sin\theta}$ with $r,\rho>0$, $\varphi,\theta\in(-\frac{\pi}{2},\frac{3\pi}{2})$, explicitly given by
\begin{subequations}\label{Eq_G}
\begin{align}
G_\D(t,\mathbf{x},\mathbf{y})&=\frac{e^{-\frac{(r+\rho)^2}{4it}}}{8i\pi t}\bigg(\Lambda\bigg(\frac{\sqrt{r\rho}\,\cos(\frac{\varphi-\theta}{2})}{\sqrt{it}}\bigg)-\Lambda\bigg(-\frac{\sqrt{r\rho}\,\sin(\frac{\varphi+\theta}{2})}{\sqrt{it}}\bigg)\bigg), \label{Eq_GD} \\
G_\N(t,\mathbf{x},\mathbf{y})&=\frac{e^{-\frac{(r+\rho)^2}{4it}}}{8i\pi t}\bigg(\Lambda\bigg(\frac{\sqrt{r\rho}\,\cos(\frac{\varphi-\theta}{2})}{\sqrt{it}}\bigg)+\Lambda\bigg(-\frac{\sqrt{r\rho}\,\sin(\frac{\varphi+\theta}{2})}{\sqrt{it}}\bigg)\bigg).
\end{align}
\end{subequations}
Here the indices $D$ and $N$ indicate the type of boundary conditions (Dirichlet or Neumann) in \eqref{Eq_Schroedinger_equation}, and for a shorter notation we used the entire function
\begin{equation}\label{Eq_Lambda}
\Lambda(z):=\frac{2}{\sqrt{\pi}}\int_0^\infty e^{-s^2-2zs}ds,\quad z\in\mathbb{C}.
\end{equation}
It can be shown that $\Lambda(z)=e^{z^2}(1-\erf(z))$ is a modification of the well-known error function. Since we allow initial conditions $F\in\mathcal{A}_1(\mathbb{C}^2)$, which may grow exponentially at $\infty$, some regularization is needed in the integral \eqref{Eq_Psi_formal}, such that the solution of \eqref{Eq_Schroedinger_equation} can be written as
\begin{equation}\label{Eq_Psi_regularized}
\Psi_\DN(t,\mathbf{x})=\lim\limits_{\varepsilon\rightarrow 0^+}\int_\Omega e^{-\varepsilon|\mathbf{y}|^2}G_\DN(t,\mathbf{x},\mathbf{y})F(\mathbf{y})d\mathbf{y},\qquad t>0,\,\mathbf{x}\in\Omega.
\end{equation}
The aim of this section is to find an absolute convergent integral representation of the wave functions in \eqref{Eq_Psi_regularized}, which will then be needed in the sequel to prove the main results of the paper. The key ingredient will be the so called Fresnel integral technique, which roughly speaking rotates the domain of integration into the complex plane and produces in this way an absolute convergent integrand. Note that the following Lemma~\ref{lem_Fresnel_integral} is a special version of \cite[Proposition 2.1]{ABCS22}, where also a proof can be found.

\begin{lem}[Fresnel integral]\label{lem_Fresnel_integral}
Let $a>0$ and $f:\mathbb{C}_{\Re>0}\rightarrow\mathbb{C}$ be holomorphic on
\begin{equation*}
\mathbb{C}_{\Re>0}:=\Set{z\in\mathbb{C} | \Re(z)>0},
\end{equation*}
and satisfies the estimate
\begin{equation}\label{Eq_Fresnel_estimate}
|f(z)|\leq Ae^{B|z|},\qquad z\in\mathbb{C}_{\Re>0},
\end{equation}
for some and $A,B\geq 0$. Then, for every $\alpha\in(0,\frac{\pi}{2})$, we get
\begin{equation*}
\lim\limits_{\varepsilon\rightarrow 0^+}\int_0^\infty e^{-\varepsilon y^2}e^{iay^2}f(y)dy=e^{i\alpha}\int_0^\infty e^{ia(ye^{i\alpha})^2}f(ye^{i\alpha})dy.
\end{equation*}
\end{lem}

The following Theorem~\ref{satz_Psi_Fresnel} now uses this Fresnel integral technique, to rewrite the regularized integral \eqref{Eq_Psi_regularized} as an absolute convergent integral in the complex plan.

\begin{satz}\label{satz_Psi_Fresnel}
Let $F\in\mathcal{A}_1(\mathbb{C}^2)$. Then, for every $\alpha\in(0,\frac{\pi}{2})$ the functions $\Psi_\D$ and $\Psi_\N$ in \eqref{Eq_Psi_regularized} can be written as the absolute convergent integral
\begin{equation}\label{Eq_Psi_Fresnel}
\Psi_\DN(t,\mathbf{x})=e^{2i\alpha}\int_0^\infty\int_{-\frac{\pi}{2}}^{\frac{3\pi}{2}}G_\DN\big(t,\mathbf{x},\rho e^{i\alpha}\vvect{\cos\theta}{\sin\theta}\big)F\big(\rho e^{i\alpha}\vvect{\cos\theta}{\sin\theta}\big)\rho\,d\theta d\rho.
\end{equation}
\end{satz}

\begin{proof}
First, we transform the integral \eqref{Eq_Psi_regularized} into polar coordinates, which is
\begin{equation}\label{Eq_PsiD_polar}
\Psi_\DN(t,\mathbf{x})=\lim\limits_{\varepsilon\rightarrow 0^+}\int_0^\infty\int_{-\frac{\pi}{2}}^{\frac{3\pi}{2}}e^{-\varepsilon\rho^2}G_\DN\big(t,\mathbf{x},\rho\vvect{\cos\theta}{\sin\theta}\big)F\big(\rho\vvect{\cos\theta}{\sin\theta}\big)\rho\,d\theta d\rho.
\end{equation}
Next, we note that the function $(0,\infty)\mapsto G_\DN(t,\mathbf{x},\rho\vect{\cos\theta}{\sin\theta})$ holomorphically extends to $\mathbb{C}_{\Re>0}$ by simply replacing $\rho\in(0,\infty)$ by $z\in\mathbb{C}_{\Re>0}$ in \eqref{Eq_G}. Note that it could also be extended to $\mathbb{C}\setminus(-\infty,0]$, which is the maximal domain of the square root in \eqref{Eq_Gtilde}, but this is not necessary for our purposes. Furthermore, we can decompose this extended Green's function into
\begin{align}
G_\DN\big(t,\mathbf{x},&z\vvect{\cos\theta}{\sin\theta}\big) \notag \\
&=e^{-\frac{z^2}{4it}}\underbrace{\frac{e^{-\frac{r^2+2rz}{4it}}}{8i\pi t}\bigg(\Lambda\bigg(\frac{\sqrt{rz}\,\cos(\frac{\varphi-\theta}{2})}{\sqrt{it}}\bigg)\mp\Lambda\bigg(-\frac{\sqrt{rz}\,\sin(\frac{\varphi+\theta}{2})}{\sqrt{it}}\bigg)\bigg)}_{=:\widetilde{G}_\DN\big(t,\mathbf{x},z\vvect{\cos\theta}{\sin\theta}\big)}, \label{Eq_Gtilde}
\end{align}
and hence write the integral \eqref{Eq_PsiD_polar} as
\begin{equation*}
\Psi_\DN(t,\mathbf{x})=\lim\limits_{\varepsilon\rightarrow 0^+}\int_0^\infty e^{-\varepsilon\rho^2}e^{-\frac{\rho^2}{4it}}\int_{-\frac{\pi}{2}}^{\frac{3\pi}{2}}\widetilde{G}_\DN\big(t,\mathbf{x},\rho\vvect{\cos\theta}{\sin\theta}\big)F\big(\rho\vvect{\cos\theta}{\sin\theta}\big)\rho\,d\theta d\rho.
\end{equation*}
The idea is now to apply the Fresnel integral technique of Lemma~\ref{lem_Fresnel_integral} with respect to the radial part of the above integral. To do so, we have to check if the function
\begin{equation}\label{Eq_Gtilde_integral}
z\mapsto\int_{-\frac{\pi}{2}}^{\frac{3\pi}{2}}\widetilde{G}_\DN\big(t,\mathbf{x},z\vvect{\cos\theta}{\sin\theta}\big)F\big(z\vvect{\cos\theta}{\sin\theta}\big)z\,d\theta
\end{equation}
is holomorphic on $\mathbb{C}_{\Re>0}$ and exponentially bounded as in \eqref{Eq_Fresnel_estimate}. It is obvious that the integrand
\begin{equation*}
z\mapsto\widetilde{G}_\DN\big(t,\mathbf{x},z\vvect{\cos\theta}{\sin\theta}\big)F\big(z\vvect{\cos\theta}{\sin\theta}\big)
\end{equation*}
is holomorphic in $\mathbb{C}_{\Re>0}$. Also knowing that
\begin{equation*}
\theta\mapsto\frac{d}{dz}\widetilde{G}_\DN\big(t,\mathbf{x},z\vvect{\cos\theta}{\sin\theta}\big)F\big(z\vvect{\cos\theta}{\sin\theta}\big),
\end{equation*}
is continuous, it follows that also the integral \eqref{Eq_Gtilde_integral} is holomorphic on $\mathbb{C}_{\Re>0}$ for every fixed $t>0$, $\mathbf{x}\in\Omega$. To verify the exponential bound \eqref{Eq_Fresnel_estimate} for the mapping \eqref{Eq_Gtilde_integral}, we first estimate the reduced Green's function $\widetilde{G}_\DN$ by
\begin{align}
\big|\widetilde{G}_\DN\big(t,\mathbf{x},z\vvect{\cos\theta}{\sin\theta}\big)\big|&=\frac{e^{-\frac{r\Im(z)}{2t}}}{8\pi t}\bigg|\Lambda\bigg(\frac{\sqrt{rz}\,\cos(\frac{\varphi-\theta}{2})}{\sqrt{it}}\bigg)\mp\Lambda\bigg(-\frac{\sqrt{rz}\,\sin(\frac{\varphi+\theta}{2})}{\sqrt{it}}\bigg)\bigg| \notag \\
&\leq\frac{e^{\frac{r|z|}{2t}}}{4\pi t}\Big(e^{\frac{r|z|}{t}\cos^2(\frac{\varphi-\theta}{2})}+e^{\frac{r|z|}{t}\sin^2(\frac{\varphi+\theta}{2})}\Big) \notag \\
&\leq\frac{1}{2\pi t}e^{\frac{3r|z|}{2t}},\quad z\in\mathbb{C}_{\Re>0}, \label{Eq_Gtilde_estimate}
\end{align}
where we used the estimate
\begin{equation*}
|\Lambda(z)|\leq\frac{2}{\sqrt{\pi}}\int_0^\infty e^{-s^2-2\Re(z)s}ds\leq\frac{2e^{|z|^2}}{\sqrt{\pi}}\int_\mathbb{R}e^{-(s+\Re(z))^2}ds=2e^{|z|^2},\quad z\in\mathbb{C},
\end{equation*}
of the function $\Lambda$ in \eqref{Eq_Lambda}. Since $F\in\mathcal{A}_1(\mathbb{C}^2)$, there exist constants $A,B\geq 0$, such that
\begin{equation}\label{Eq_F_estimate}
\big|F\big(z\vvect{\cos\theta}{\sin\theta}\big)\big|\leq Ae^{B\sqrt{|z\cos\theta|^2+|z\sin\theta|^2}}=Ae^{B|z|},\quad z\in\mathbb{C}.
\end{equation}
Combining now \eqref{Eq_Gtilde_estimate}, \eqref{Eq_F_estimate} as well as $|z|\leq e^{|z|}$, which is an immediate consequence of the power series expansion of the exponential, the integral in \eqref{Eq_Gtilde_integral} admits the estimate
\begin{equation}\label{Eq_Integral_Gtilde_estimate}
\bigg|\int_{-\frac{\pi}{2}}^{\frac{3\pi}{2}}\widetilde{G}_\DN\big(t,\mathbf{x},z\vvect{\cos\theta}{\sin\theta}\big)F\big(z\vvect{\sin\theta}{\cos\theta}\big)z\,d\theta\bigg|\leq\frac{A}{t}e^{(\frac{3r}{2t}+B+1)|z|},\quad z\in\mathbb{C}_{\Re>0}.
\end{equation}
Hence, we verified that the assumptions of Lemma~\ref{lem_Fresnel_integral} for the mapping \eqref{Eq_Gtilde_integral} are satisfied  and so we can write the integral \eqref{Eq_PsiD_polar} for any $\alpha\in(0,\frac{\pi}{2})$ in the form
\begin{equation*}
\Psi_\DN(t,\mathbf{x})=e^{i\alpha}\int_0^\infty e^{-\frac{(\rho e^{i\alpha})^2}{4it}}\int_{-\frac{\pi}{2}}^{\frac{3\pi}{2}}\widetilde{G}_\DN\big(t,\mathbf{x},\rho e^{i\alpha}\vvect{\cos\theta}{\sin\theta}\big)F\big(\rho e^{i\alpha}\vvect{\sin\theta}{\cos\theta}\big)\rho e^{i\alpha}d\theta d\rho,
\end{equation*}
which, after substituting the defintion of $\widetilde{G}_\DN$ from \eqref{Eq_Gtilde}, is exactly the stated representation \eqref{Eq_Psi_Fresnel}.
\end{proof}

\section{The infinite order differential operators associated \\ with  half-plane barrier }\label{sec_The_time_evolution_operator}

In this section we introduce, based on the Green's function integral \eqref{Eq_Psi_Fresnel}, another representation of the solution $\Psi_\DN(t,\mathbf{x})$ of \eqref{Eq_Schroedinger_equation}, using some infinite order differential operator acting on the initial condition $F$. More precisely, we use the two-dimensional power series representation
\begin{equation*}
F(\mathbf{z})=\sum\limits_{n_1,n_2=0}^\infty\frac{\partial_{z_1}^{n_1}\partial_{z_2}^{n_2}F(\mathbf{0})}{n_1!n_2!}z_1^{n_1}z_2^{n_2},\quad\mathbf{z}=\vvect{z_1}{z_2}\in\mathbb{C}^2,
\end{equation*}
to rewrite (for the moment formally) the function $\Psi_\DN(t,\mathbf{x})$ in \eqref{Eq_Psi_Fresnel} as
\begin{align}
\Psi_\DN(t,\mathbf{x})&=e^{2i\alpha}\int_0^\infty\int_{-\frac{\pi}{2}}^{\frac{3\pi}{2}}G_\DN\big(t,\mathbf{x},\rho e^{i\alpha}\vvect{\cos\theta}{\sin\theta}\big) \notag \\
&\hspace{3.5cm}\times\sum\limits_{n_1,n_2=0}^\infty\frac{\partial_{z_1}^{n_1}\partial_{z_2}^{n_2}F(\mathbf{0})}{n_1!n_2!}(\rho e^{i\alpha}\cos\theta)^{n_1}(\rho e^{i\alpha}\sin\theta)^{n_2}\rho\,d\theta d\rho \notag \\
&=\sum\limits_{n_1,n_2=0}^\infty\frac{e^{(n_1+n_2+2)i\alpha}}{n_1!n_2!}\int_0^\infty\int_{-\frac{\pi}{2}}^{\frac{3\pi}{2}}G_\DN\big(t,\mathbf{x},\rho e^{i\alpha}\vvect{\cos\theta}{\sin\theta}\big) \notag \\
&\hspace{3.5cm}\times(\cos\theta)^{n_1}(\sin\theta)^{n_2}\rho^{n_1+n_2+1}d\theta d\rho\frac{\partial^{n_1+n_2}}{\partial z_1^{n_1}\partial z_2^{n_2}}F(\mathbf{z})\Big|_{\mathbf{z}=\mathbf{0}} \notag \\
&=\sum\limits_{n_1,n_2=0}^\infty c_{n_1,n_2}(t,\mathbf{x})\frac{\partial^{n_1+n_2}}{\partial z_1^{n_1}\partial z_2^{n_2}}F(\mathbf{z})\Big|_{\mathbf{z}=\mathbf{0}}, \label{Eq_U_derivation}
\end{align}
using the coefficients
\begin{equation}\label{Eq_c}
c_{n_1,n_2}(t,\mathbf{x}):=\frac{e^{(n_1+n_2+2)i\alpha}}{n_1!n_2!}\int_0^\infty\int_{-\frac{\pi}{2}}^{\frac{3\pi}{2}}G_\DN\big(t,\mathbf{x},\rho e^{i\alpha}\vvect{\cos\theta}{\sin\theta}\big)(\cos\theta)^{n_1}(\sin\theta)^{n_2}\rho^{n_1+n_2+1}d\theta d\rho.
\end{equation}
The above computations mean, that using the \textit{infinite order differential operator}
\begin{equation}\label{Eq_U}
U_\DN(t,\mathbf{x}):=\sum\limits_{n_1,n_2=0}^\infty c_{n_1,n_2}(t,\mathbf{x})\frac{\partial^{n_1+n_2}}{\partial z_1^{n_1}\partial z_2^{n_2}},
\end{equation}
we can write the solution $\Psi_\DN(t,\mathbf{x})$  as
\begin{equation}\label{Eq_Psi_U}
\Psi_\DN(t,\mathbf{x})=U_\DN(t,\mathbf{x})F(z)\Big|_{z=0},\qquad t>0,\,\mathbf{x}\in\Omega.
\end{equation}
The main advantage of the representation \eqref{Eq_Psi_U} using the operator $U_\DN(t,\mathbf{x})$ will be, that the many properties $\Psi_\DN$, as the continuous dependency result  or the supershift property discussed in the sequel, will turn out to be simple consequences of the continuity of this operator in the space $\mathcal{A}_1(\mathbb{C}^2)$. However, in order to prove this continuity and also to make the calculations in \eqref{Eq_U_derivation} rigorous, we need the following lemma about the exponential boundedness of the derivatives of functions in $\mathcal{A}_1(\mathbb{C}^2)$.

\begin{lem}
If a function $F\in\mathcal{A}^1(\mathbb{C}^2)$ satisfies the estimate $|F(\mathbf{z})|\leq Ae^{B|\mathbf{z}|}$, for some $A\geq 0$, $B>0$, then for every $n_1,n_2\in\mathbb{N}_0$, the derivatives of $F$ can be estimated as
\begin{equation}\label{Eq_F_derivative_estimate}
\Big|\frac{\partial^{n_1+n_2}}{\partial z_1^{n_1}\partial z_2^{n_2}}F(\mathbf{z})\Big|\leq A(eB)^{n_1+n_2}e^{B|\mathbf{z}|},\qquad\mathbf{z}=\vvect{z_1}{z_2}\in\mathbb{C}^2.
\end{equation}
\end{lem}

\begin{proof}
Let us first consider the case $n_1,n_2\neq 0$. By the Cauchy-formula we can represent the derivative $\frac{\partial^{n_1+n_2}}{\partial z_1^{n_1}z_2^{n_2}}F(\mathbf{z})$, for every fixed $\mathbf{z}\in\mathbb{C}^2$ by the integral
\begin{equation}\label{Eq_F_derivative_estimate_2}
\frac{\partial^{n_1+n_2}}{\partial z_1^{n_1}z_2^{n_2}}F(\mathbf{z})=\frac{n_1!n_2!}{(2\pi i)^2}\int_{|\xi_2-z_2|=r_2}\int_{|\xi_1-z_1|=r_1}\frac{F\vvect{\xi_1}{\xi_2}}{(\xi_1-z_1)^{n_1+1}(\xi_2-z_2)^{n_2+1}}d\xi_1d\xi_2,
\end{equation}
where $r_1,r_2>0$ for the moment arbitrary and will be specified later. Hence we can estimate the derivative by
\begin{align}
\Big|\frac{\partial^{n_1+n_2}}{\partial z_1^{n_1}z_2^{n_2}}F(\mathbf{z})\Big|&=\frac{n_1!n_2!}{4\pi^2}\bigg|\int_0^{2\pi}\int_0^{2\pi}\frac{F\vvect{z_1+r_1e^{i\varphi_1}}{z_2+r_2e^{i\varphi_2}}}{(r_1e^{i\varphi_1})^{n_1}(r_2e^{i\varphi_2})^{n_2}}d\varphi_1d\varphi_2\bigg| \notag \\
&\leq\frac{An_1!n_2!}{4\pi^2r_1^{n_1}r_2^{n_2}}\int_0^{2\pi}\int_0^{2\pi}e^{B\sqrt{|z_1+r_1e^{i\varphi_1}|^2+|z_2+r_2e^{i\varphi_2}|^2}}d\varphi_1d\varphi_2 \notag \\
&\leq\frac{An_1!n_2!}{r_1^{n_1}r_2^{n_2}}e^{B\sqrt{(|z_1|+r_1)^2+(|z_2|+r_2)^2}} \notag \\
&\leq\frac{An_1!n_2!}{r_1^{n_1}r_2^{n_2}}e^{B\sqrt{|z_1|^2+|z_2|^2}+B\sqrt{r_1^2+r_2^2}}. \label{Eq_F_derivative_estimate_1}
\end{align}
Assuming now for the moment $z_1,z_2\neq 0$. Then we choose the radii $r_1,r_2$ as
\begin{equation}\label{Eq_r_definition}
r_1:=\frac{(n_1!)^{\frac{1}{n_1}}}{B}\qquad\text{and}\qquad r_2:=\frac{(n_2!)^{\frac{1}{n_2}}}{B}.
\end{equation}
Using the inequality $n!\leq n^n$ for every $n\geq 1$, we can then estimate these radii by
\begin{equation}\label{Eq_r_estimate}
r_1\leq\frac{n_1}{B}\qquad\text{and}\qquad r_2\leq\frac{n_2}{B}.
\end{equation}
Using now the values \eqref{Eq_r_definition} of $r_1$ and $r_2$ in the denomenator of \eqref{Eq_F_derivative_estimate_1} and the estimates \eqref{Eq_r_estimate} in the exponent of \eqref{Eq_F_derivative_estimate_1}, leads to the stated estimate
\begin{equation*}
\Big|\frac{\partial^{n_1+n_2}}{\partial z_1^{n_1}z_2^{n_2}}F(\mathbf{z})\Big|\leq A(eB)^{n_1+n_2}e^{B(|z_1|+|z_2|)},
\end{equation*}
whenever $z_1,z_2\neq 0$. However, since both sides of this inequality are continuous functions in $z_1,z_2$, it can be extended to every $z_1,z_2\in\mathbb{C}$ by continuity. The case $n_1=0$ and/or $n_2=0$ follows the same calculations with the choice $r_1=1$ and/or $r_2=1$ in \eqref{Eq_F_derivative_estimate_1}.
er reduces to
\begin{equation*}
\Big|\frac{\partial^{n_1+n_2}}{\partial z_1^{n_1}z_2^{n_2}}F(\mathbf{z})\Big|\leq A(eB)^{n_1+n_2}e^{B|\mathbf{z}|}.
\end{equation*}
For the case where both $n_1=n_2=0$ vanish, the estimate \eqref{Eq_F_derivative_estimate} is trivial. If exactly one of the numbers $n_1,n_2$ vanish, lets say $n_1\neq 0$ and $n_2=0$, we can do a similar computation, only replacing the Cauchy-formula \eqref{Eq_F_derivative_estimate_2} by
\begin{equation*}
\frac{\partial^{n_1}}{\partial z_1^{n_1}}F(\mathbf{z})=\frac{n_1!}{2\pi i}\int_{|\xi_1-z_1|=r_1}\frac{F\vvect{\xi_1}{z_2}}{(\xi_1-z_1)^{n_1+1}}d\xi_1. \qedhere
\end{equation*}
\end{proof}

Next we prove the main property of the operator $U_\DN(t,\mathbf{x})$, being a continuous operator in the space $\mathcal{A}_1(\mathbb{C})$.

\begin{satz}\label{satz_U_continuity}
For every fixed $t>0$, $\mathbf{x}\in\Omega$, the operator $U_\DN(t,\mathbf{x}):\mathcal{A}_1(\mathbb{C}^2)\rightarrow\mathcal{A}_1(\mathbb{C}^2)$ is continuous. Moreover, there exists some constant $C(t,\mathbf{x})\geq 0$, continuously dependin on $t$ and $\mathbf{x}$, such that
\begin{equation}\label{Eq_U_continuity}
|U_\DN(t,\mathbf{x})F(z)|\leq AC(t,\mathbf{x})e^{B|\mathbf{z}|},\qquad\mathbf{z}\in\mathbb{C}^2,
\end{equation}
whenever $F\in\mathcal{A}_1(\mathbb{C}^2)$ satisfies $|F(\mathbf{z})|\leq Ae^{B|\mathbf{z}|}$ for some $A\geq 0$, $B>0$.
\end{satz}

\begin{proof}
Using the decomposition \eqref{Eq_Gtilde} of the Green's function $G_\DN(t,\mathbf{x},\mathbf{z})$ and the estimate \eqref{Eq_Gtilde_estimate} of the reduced Green's function $\widetilde{G}(t,\mathbf{x},\mathbf{z})$, we can estimate the coefficients \eqref{Eq_c} by
\begin{align}
|c_{n_1,n_2}(t,\mathbf{x})|&\leq\frac{1}{2\pi tn_1!n_2!}\int_0^\infty\int_{-\frac{\pi}{2}}^{\frac{3\pi}{2}}e^{-\frac{\rho^2\sin(2\alpha)}{4t}}e^{\frac{3r\rho}{2t}}|\cos\theta|^{n_1}|\sin\theta|^{n_2}\rho^{n_1+n_2+1}d\theta d\rho \notag \\
&\leq\frac{1}{tn_1!n_2!}\int_0^\infty e^{-\frac{\rho^2\sin(2\alpha)}{4t}}e^{\frac{3r\rho}{2t}}\rho^{n_1+n_2+1}d\rho \notag \\
&=\frac{1}{tn_1!n_2!}\Big(\frac{8t}{\sin(2\alpha)}\Big)^{\frac{n_1+n_2+2}{2}}\int_0^\infty e^{-2\rho^2+\frac{3\sqrt{2}\,r\rho}{\sqrt{t\sin(2\alpha)}}}\rho^{n_1+n_2+1}d\rho \notag \\
&\leq\frac{1}{tn_1!n_2!}\Big(\frac{8t}{\sin(2\alpha)}\Big)^{\frac{n_1+n_2+2}{2}}e^{\frac{9r^2}{2t\sin(2\alpha)}}\int_0^\infty e^{-\rho^2}\rho^{n_1+n_2+1}d\rho \notag \\
&=\frac{\Gamma(\frac{n_1+n_2+2}{2})}{2tn_1!n_2!}\Big(\frac{8t}{\sin(2\alpha)}\Big)^{\frac{n_1+n_2+2}{2}}e^{\frac{9r^2}{2t\sin(2\alpha)}} \notag \\
&\leq\frac{\pi^2}{2t\Gamma(\frac{n_1+1}{2})\Gamma(\frac{n_2+1}{2})}\Big(\frac{16t}{\sin(2\alpha)}\Big)^{\frac{n_1+n_2+2}{2}}e^{\frac{9r^2}{2t\sin(2\alpha)}}, \label{Eq_c_estimate}
\end{align}
where we used the estimates
\begin{align*}
&\Gamma(a+b+1)\leq 2^{a+b+1}\Gamma\Big(a+\frac{1}{2}\Big)\Gamma\Big(b+\frac{1}{2}\Big)\quad\text{and} \\
&\frac{\Gamma(\frac{a+1}{2})^2}{\Gamma(a+1)}=B\Big(\frac{a+1}{2},\frac{a+1}{2}\Big)\leq B\Big(\frac{1}{2},\frac{1}{2}\Big)=\pi,
\end{align*}
of the $\Gamma$-function, which are true for every $a,b\geq 0$. Hence, if we assume that $|F(\mathbf{z})|\leq A^{B|\mathbf{z}|}$, we can use the estimate \eqref{Eq_F_derivative_estimate} of the derivatives of $F$ to estimate the action of the operator $U_\DN$ by
\begin{align*}
\big|U_\DN(t,\mathbf{x})F(\mathbf{z})\big|&=\bigg|\sum\limits_{n_1,n_2=0}^\infty c_{n_1,n_2}(t,\mathbf{x})\frac{\partial^{n_1+n_2}}{\partial z_1^{n_1}\partial z_2^{n_2}}F(\mathbf{z})\bigg| \\
&\leq\frac{8A\pi^2}{\sin(2\alpha)}e^{\frac{9r^2}{2t\sin(2\alpha)}}\sum\limits_{n_1,n_2=0}^\infty\frac{1}{\Gamma(\frac{n_1+1}{2})\Gamma(\frac{n_2+1}{2})}\Big(\frac{4eB\sqrt{t}}{\sqrt{\sin(2\alpha)}}\Big)^{n_1+n_2}e^{B|\mathbf{z}|} \\
&=\frac{8A\pi^2}{\sin(2\alpha)}e^{\frac{9r^2}{2t\sin(2\alpha)}}E_{\frac{1}{2},\frac{1}{2}}\Big(\frac{4eB\sqrt{t}}{\sqrt{\sin(2\alpha)}}\Big)^2e^{B|\mathbf{z}|},\qquad\mathbf{z}\in\mathbb{C}^2.
\end{align*}
Hence $U_\DN(t,\mathbf{x})F\in\mathcal{A}_1(\mathbb{C}^2)$ and the inequality \eqref{Eq_U_continuity} is satisfied with the constant
\begin{equation*}
C(t,\mathbf{x})=\frac{8\pi^2}{\sin(2\alpha)}e^{\frac{9r^2}{2t\sin(2\alpha)}}E_{\frac{1}{2},\frac{1}{2}}\Big(\frac{4eB\sqrt{t}}{\sqrt{\sin(2\alpha)}}\Big)^2.
\end{equation*}
For the proof of the continuity let $F,(F_n)_{n\in\mathbb{N}}\in\mathcal{A}_1(\mathbb{C}^2)$ such that $\lim\limits_{n\rightarrow\infty}F_n=F$ in $\mathcal{A}_1(\mathbb{C}^2)$. By \eqref{Eq_A1_convergence} this means that there exists some $B\geq 0$ such that
\begin{equation*}
A_n:=\sup\limits_{z\in\mathbb{C}^2}|F_n(z)-F(z)|e^{-B|z|}\overset{n\rightarrow\infty}{\longrightarrow}0.
\end{equation*}
With this constants $A_n$, the difference admits the estimate $|F_n(\mathbf{z})-F(\mathbf{z})|\leq A_ne^{B|\mathbf{z}|}$ and using \eqref{Eq_U_continuity}, we get
\begin{equation}\label{Eq_Continuity_estimate}
\sup\limits_{\mathbf{z}\in\mathbb{C}^2}\big|U_\DN(t,\mathbf{x})F_n(\mathbf{z})-U_\DN(t,\mathbf{x})F(\mathbf{z})\big|e^{-B|\mathbf{z}|}\leq A_nC(t,\mathbf{x})\overset{n\rightarrow\infty}{\longrightarrow}0.
\end{equation}
This proves the convergence $\lim_{n\rightarrow\infty}U_\DN(t,\mathbf{x})F_n=U_\DN(t,\mathbf{x})F$ in $\mathcal{A}_1(\mathbb{C}^2)$ and hence the continuity of $U_\DN(t,\mathbf{x})$.
\end{proof}

The next lemma uses the operator $U_\DN(t,\mathbf{x})$ to prove the continuous dependency of the solution $\Psi_\DN$ from the initial value $F$. In order to emphasize the initial value, we will use the notation $\Psi_\DN(t,\mathbf{x};F)$ in the following.

\begin{cor}\label{cor_Continuous_dependency}
Let $F,(F_n)_{n\in\mathbb{N}}\in\mathcal{A}_1(\mathbb{C}^2)$ be such that $\lim_{n\rightarrow\infty}F_n=F$ in $\mathcal{A}_1(\mathbb{C}^2)$. Then

\begin{enumerate}
\item[i)] $\Psi_\DN(t,\mathbf{x};F)=U_\DN(t,\mathbf{x})F(\mathbf{z})\big|_{\mathbf{z}=\mathbf{0}}$,
\item[ii)] $\lim\limits_{n\rightarrow\infty}\Psi_\DN(t,\mathbf{x};F_n)=\Psi_\DN(t,\mathbf{x};F)$\quad uniformly on compact subsets of $(0,\infty)\times\Omega$.
\end{enumerate}
\end{cor}

\begin{proof}
The representation i) follows from the calculations in \eqref{Eq_U_derivation}, where the differentiation $\frac{\partial^{n_1+n_2}}{\partial z_1^{n_1}\partial z_2^{n_2}}$ and summation $\sum_{n_1,n_2=0}^\infty$ can be carried outside the integral due to the estimate
\begin{align*}
\bigg|G_\DN\big(t,\mathbf{x},\rho e^{i\alpha}\vvect{\cos\theta}{\sin\theta}\big)\frac{\partial_{z_1}^{n_1}\partial_{z_2}^{n_2}F(\mathbf{0})}{n_1!n_2!}&(\rho e^{i\alpha}\cos\theta)^{n_1}(\rho e^{i\alpha}\sin\theta)^{n_2}\rho\bigg| \\
&\leq\frac{A(eB)^{n_1+n_2}}{2\pi tn_1!n_2!}e^{-\frac{\rho^2\sin(2\alpha)}{4t}}e^{\frac{3r\rho}{2t}}\rho^{n_1+n_2+1},
\end{align*}
following from \eqref{Eq_Gtilde_estimate} and \eqref{Eq_F_derivative_estimate}, which makes
\begin{equation*}
\sum\limits_{n_1,n_2=0}^\infty\int_0^\infty\int_{-\frac{\pi}{2}}^{\frac{3\pi}{2}}\bigg|G_\DN\big(t,\mathbf{x},\rho e^{i\alpha}\vvect{\cos\theta}{\sin\theta}\big)\frac{\partial_{z_1}^{n_1}\partial_{z_2}^{n_2}F(\mathbf{0})}{n_1!n_2!}(\rho e^{i\alpha}\cos\theta)^{n_1}(\rho e^{i\alpha}\sin\theta)^{n_2}\rho\bigg|d\varphi d\rho<\infty
\end{equation*}
absolute convergent. In order to prove the convergence in ii), we note that by i) and \eqref{Eq_Continuity_estimate} we get
\begin{equation*}
\big|\Psi_\DN(t,\mathbf{x};F_n)-\Psi_\DN(t,\mathbf{x};F)\big|
=\Big|U_\DN(t,\mathbf{x})(F_n(\mathbf{z})-F(\mathbf{z}))\Big|_{\mathbf{z}=\mathbf{0}}\leq A_nC(t,\mathbf{x}),
\end{equation*}
using the coefficients 
$$
A_n:=\sup_{z\in\mathbb{C}^2}|F_n(z)-F(z)|e^{-B|z|}\overset{n\rightarrow\infty}{\longrightarrow}0.
$$
 Since the coefficient $C(t,\mathbf{x})$ is moreover continuous by Theorem \ref{satz_U_continuity}, the convergence 
 $$
 \lim_{n\rightarrow\infty}\Psi_\DN(t,\mathbf{x};F_n)=\Psi_\DN(t,\mathbf{x};F)
 $$
  is uniform on compact subsets of $(0,\infty)\times\Omega$.
\end{proof}

\section{Time persistence of superoscillations and  supershift}\label{sec_Time_persistence_of_superoscillations}

In this section, we investigate the evolution of superoscillating functions as initial conditions in the Schrödinger equation \eqref{Eq_Schroedinger_equation}. The expectation that for any sequence $(F_n)_n$ of superoscillating initial conditions, the sequence of solutions $\Psi_\DN(t,\mathbf{x};F_n)$ will again be superoscillating (for fixed times $t>0$), can easily be negated. In fact, one way of reasoning is that the convergence \eqref{Eq_Fn_convergence} of the initial conditions $F_n$ to a plane wave $e^{iax}$ implies that the solutions
 converge as
\begin{equation*}
\lim\limits_{n\rightarrow\infty}\Psi_\DN(t,\mathbf{x};F_n)=\Psi_\DN\big(t,\mathbf{x};e^{i\mathbf{a}\,\cdot\,}\big),
\end{equation*}
see, e.g., the continuous dependency result in Corollary~\ref{cor_Continuous_dependency}~ii). However, for $\Psi_\DN(t,\mathbf{x};F_n)$ to be superoscillating, the limit function $\Psi_\DN(t,\mathbf{x};e^{i\mathbf{a}\,\cdot\,})$ has to be plane wave due to the definition in \eqref{Eq_Fn_convergence}, which is not possible since 
 the boundary condition forces the wave function (or its derivative) 
 to vanish on $\Gamma$. 
 Moreover, the solution $\Psi_\DN(t,\mathbf{x};e^{i\mathbf{a}\,\cdot\,})$ will no longer be a holomorphic function in the $\mathbf{x}$-variable and hence no element in the space $\mathcal{A}_1(\mathbb{C}^2)$. These considerations show that the precise mathematical notion of superoscillations is too narrow to persist in time.

\medskip

This motivates the following notion of \textit{supershift}, which basically is a replacement of the holomorphic exponentials $e^{i\mathbf{ax}}$ by arbitrary continuous functions $\varphi_\mathbf{a}(\mathbf{x})$ in the Definition \ref{defi_Superoscillations} of superoscillations.

\begin{defi}[Supershift]\label{defi_Supershift}
Let $X$ be a metric space and
\begin{equation}\label{Eq_varphik}
\varphi_\mathbf{k}:X\rightarrow\mathbb{C},\qquad\mathbf{k}\in\mathbb{C}^2,
\end{equation}
be a family of complex valued functions such that $\mathbf{k}\mapsto\varphi_\mathbf{k}(s)$ is continuous for every $s\in X$. We say that a sequence of the form
\begin{equation}\label{Eq_Phin_integral}
\Phi_n(s):=\int_{|\mathbf{k}|\leq k_0}\varphi_\mathbf{k}(s)d\mu_n(\mathbf{k}),\qquad s\in X,
\end{equation}
for some $k_0>0$ and complex Borel measures $\mu_n$ on the closed ball $\overline{B_{k_0}(0)}\subseteq\mathbb{C}^2$, admits a \textit{supershift}, if there exists some $\mathbf{a}\in\mathbb{C}^2$ with $|\mathbf{a}|>k_0$, such that
\begin{equation}\label{Eq_Phin_convergence}
\lim\limits_{n\rightarrow\infty}\Phi_n(s)=\varphi_\mathbf{a}(s),\qquad s\in X,
\end{equation}
converges uniformly on compact subsets of $X$.
\end{defi}

\begin{bem}
With the special choice $X=\mathbb{C}^2$ and $\varphi_\mathbf{k}(\mathbf{z})=e^{i\mathbf{kz}}$, it turns out that the notion of supershift in Definition~\ref{defi_Supershift} is a generalization of the notion of superoscillations in Definition~\ref{defi_Superoscillations}. Indeed, with this choice the integrals \eqref{Eq_Fn_integral} and \eqref{Eq_Phin_integral} coincide and since the $\mathcal{A}_1$-convergence \eqref{Eq_A1_convergence} is stronger than the convergence on compact sets, the convergence \eqref{Eq_Fn_convergence} implies the convergence \eqref{Eq_Phin_convergence}.
\end{bem}

In the following theorem we now prove that after the interaction of superoscillations with the half-plane barrier, the superoscillatory property turns into a supershift property, which then persists for all times $t>0$.

\begin{satz}
Let $(F_n)_{n\in\mathbb{N}}$ be a superoscillating sequence according to Definition~\ref{defi_Superoscillations}, i.e.
\begin{equation}\label{Eq_Fn_time_persistence}
F_n(\mathbf{z})=\int_{|\mathbf{k}|\leq k_0}e^{i\mathbf{kz}}d\mu_n(\mathbf{k})\overset{n\rightarrow\infty}{\longrightarrow}e^{i\mathbf{az}},\quad\text{in }\mathcal{A}_1(\mathbb{C}^2).
\end{equation}
Then the sequence $\Psi_DN(t,\mathbf{x};F_n)$, $n\in\mathbb{N}$ of solutions admits a supershift according to Definition~\ref{defi_Supershift}. In particular we have
\begin{equation}\label{Eq_Psi_convergence}
\Psi_\DN(t,\mathbf{x};F_n)=\int_{|\mathbf{k}|\leq k_0}\Psi_\DN\big(t,\mathbf{x};e^{i\mathbf{k}\,\cdot\,}\big)d\mu_n(\mathbf{k})\overset{n\rightarrow\infty}{\longrightarrow}\Psi_\DN\big(t,\mathbf{x},e^{i\mathbf{a}\,\cdot\,}\big),
\end{equation}
where the convergence is uniform on any compact subset of $(0,\infty)\times\Omega$.
\end{satz}

\begin{proof}
For the first identity in \eqref{Eq_Psi_convergence}, we use the representation of the wave function via the infinite order differential operator in Corollary~\ref{cor_Continuous_dependency}~i). Then we get
\begin{align}
\Psi_\DN(t,\mathbf{x};F_n)&=U_\DN(t,\mathbf{x})\int_{|\mathbf{k}|\leq k_0}e^{i\mathbf{kz}}d\mu_n(\mathbf{k})\Big|_{\mathbf{z}=\mathbf{0}} \notag \\
&=\sum\limits_{n_1,n_2=0}^\infty c_{n_1,n_2}(t,\mathbf{x})\frac{\partial^{n_1+n_2}}{\partial z_1^{n_1}\partial z_2^{n_2}}\int_{|\mathbf{k}|\leq k_0}e^{i\mathbf{kz}}d\mu_n(\mathbf{k})\Big|_{\mathbf{z}=\mathbf{0}} \notag \\
&=\int_{|\mathbf{k}|\leq k_0}\sum\limits_{n_1,n_2=0}^\infty c_{n_1,n_2}(t,\mathbf{x})\frac{\partial^{n_1+n_2}}{\partial z_1^{n_1}\partial z_2^{n_2}}e^{i\mathbf{k}\mathbf{z}}\Big|_{\mathbf{z}=\mathbf{0}}d\mu_n(\mathbf{k}) \notag \\
&=\int_{|\mathbf{k}|\leq k_0}U_\DN(t,\mathbf{x})e^{i\mathbf{kz}}d\mu_n(\mathbf{k})=\int_{|\mathbf{k}|\leq k_0}\Psi_\DN\big(t,\mathbf{x};e^{i\mathbf{k}\,\cdot\,}\big)d\mu_n(\mathbf{k}). \label{Eq_Psi_supershift_representation}
\end{align}
Here, in the third equation we were allowed to interchange the sum and the derivative with the integral because from \eqref{Eq_c_estimate} we conclude the estimate
\begin{align*}
\Big|c_{n_1,n_2}(t,\mathbf{x})\frac{\partial^{n_1+n_2}}{\partial z_1^{n_1}\partial z_2^{n_2}}e^{i\mathbf{kz}}\Big|&=\Big|c_{n_1,n_2}(t,\mathbf{x})(ik_1)^{n_1}(ik_2)^{n_2}e^{i\mathbf{kz}}\Big| \\
&\leq\frac{\pi^2}{2t\Gamma(\frac{n_1+1}{2})\Gamma(\frac{n_2+1}{2})}\Big(\frac{16t}{\sin(2\alpha)}\Big)^{\frac{n_1+n_2+2}{2}}e^{\frac{9r^2}{2t\sin(2\alpha)}}|k_1|^{n_1}|k_2|^{n_2}e^{|\mathbf{kz}|} \\
&\leq\frac{8\pi^2}{\sin(2\alpha)\Gamma(\frac{n_1+1}{2})\Gamma(\frac{n_2+1}{2})}\Big(\frac{4k_0\sqrt{t}}{\sqrt{\sin(2\alpha)}}\Big)^{n_1+n_2}e^{\frac{9r^2}{2t\sin(2\alpha)}}e^{k_0|\mathbf{z}|},
\end{align*}
and hence the sum 
$$
\sum_{n_1,n_2=0}^\infty\big|c_{n_1,n_2}(t,\mathbf{x})\frac{\partial^{n_1+n_2}}{\partial z_1^{n_1}\partial z_2^{n_2}}e^{i\mathbf{kz}}\big|<\infty
$$
 is absolute convergent and interchanging sum and integral is allowed by the dominated convergence theorem and the fact that the measure $\mu_n$ is complex and hence finite.

\medskip

Secondly, the convergence in \eqref{Eq_Psi_convergence} has already been proven in Corollary~\ref{cor_Continuous_dependency}~ii). Since finally the representation \eqref{Eq_Psi_supershift_representation} is exactly the one of the supershift \eqref{Eq_Phin_integral} with the metric space $X=(0,\infty)\times\Omega$ and the functions
\begin{equation*}
\varphi_\mathbf{k}(t,\mathbf{x}):=\Psi_\DN\big(t,\mathbf{x};e^{i\mathbf{k}\,\cdot\,}\big). \qedhere
\end{equation*}
\end{proof}

In the following we consider the special case of superoscillating functions of the form \eqref{Eq_Fn_example} and show that the resulting wave functions admit a supershift.

\begin{cor}
Let $F_n$ be functions of the form
\begin{equation*}
F_n(\mathbf{z})=\sum\limits_{j=0}^nC_j(n)e^{i\mathbf{k_j}(n)\mathbf{z}},\quad\mathbf{z}\in\mathbb{C}^2,
\end{equation*}
with coefficients $C_j(n)\in\mathbb{C}$ and wave vectors $\mathbf{k_j}(n)\in\mathbb{R}^2$ satisfying $|\mathbf{k_j}(n)|\leq 1$. 
If 
$$
\lim_{n\rightarrow\infty}F_n(\mathbf{z})=e^{i\mathbf{az}}
$$
 converges in $\mathcal{A}_1(\mathbb{C}^2)$, for some $\mathbf{a}\in\mathbb{R}^2$ with $|\mathbf{a}|>1$, then also the sequence of solutions $\Psi_\DN(t,\mathbf{x};F_n)$ converge as
\begin{equation}\label{Eq_Psin_example}
\lim\limits_{n\rightarrow\infty}\Psi_\DN(t,\mathbf{x};F_n)=\lim\limits_{n\rightarrow\infty}\sum\limits_{j=0}^nC_j(n)\Psi_\DN\big(t,\mathbf{x};e^{i\mathbf{k_j}(n)\,\cdot\,}\big)=\Psi_\DN\big(t,\mathbf{x};e^{i\mathbf{a}\,\cdot\,}\big),
\end{equation}
uniformly for $(t,\mathbf{x})$ in compact subsets of $(0,\infty)\times\Omega$.
\end{cor}

For the final part of this paper let us now fix $t>0$ and $\mathbf{x}\in\Omega$ and look at equation \eqref{Eq_Psin_example} in terms of the mapping $\mathbf{k}\mapsto\Psi_\DN(t,\mathbf{x};e^{i\mathbf{k}\,\cdot\,})$. Then, one sees that the value of this mapping at a point $\mathbf{a}$ with $|\mathbf{a}|>1$, which is located outside the unit ball is determined by only values $|\mathbf{k_j}(n)|\leq 1$ inside the unit ball. This property looks very much like analyticity. The following proposition shows that this is indeed the case.

\begin{prop}
For every fixed $t>0$, $\mathbf{x}\in\Omega$, the mapping
\begin{equation*}
\mathbf{k}\mapsto\Psi_\DN\big(t,\mathbf{x};e^{i\mathbf{k}\,\cdot\,}\big)\quad\text{is holomorphic on }\mathbb{C}^2.
\end{equation*}
\end{prop}

\begin{proof}
Using the representation of the wave function using the infinite order differential operator in Corollary~\ref{cor_Continuous_dependency}~i), gives
\begin{align*}
\Psi_\DN\big(t,\mathbf{x};e^{i\mathbf{k}\,\cdot\,}\big)
&=U_\DN(t,\mathbf{x})e^{i\mathbf{kz}}\Big|_{\mathbf{z}=\mathbf{0}}
\\
&=\sum\limits_{n_1,n_2=0}^\infty c_{n_1,n_2}(t,\mathbf{x})\frac{\partial^{n_1+n_2}}{\partial z_1^{n_1}\partial z_2^{n_2}}e^{i\mathbf{kz}}\Big|_{\mathbf{z}=\mathbf{0}} 
\\
&=\sum\limits_{n_1,n_2=0}^\infty c_{n_1,n_2}(t,\mathbf{x})(ik_1)^{n_1}(ik_2)^{n_2}.
\end{align*}
Since this is an everywhere convergent power series in $\mathbf{k}$, the mapping $\mathbf{k}\mapsto\Psi_\DN(t,\mathbf{x};e^{i\mathbf{k}\,\cdot\,})$ is holomorphic on $\mathbb{C}^2$.
\end{proof}

\bibliographystyle{amsplain}

\end{document}